\documentclass{article}

\usepackage{amssymb,amsfonts,amsthm}
\usepackage{graphicx}
\usepackage{subfigure}
\usepackage{amssymb}
\usepackage{amsthm}
\usepackage{amsmath}
\usepackage{bm}             
\usepackage[mathscr]{eucal}
\usepackage{color}
\usepackage{cancel}
\usepackage{enumerate}
\usepackage{psfrag}

\usepackage{fullpage}

\newcommand{\R}{\mathbb{R}}

\DeclareMathOperator{\tr}{Tr}

\theoremstyle{plain}
\newtheorem{Theorem}{Theorem}
\newtheorem{Corollary}{Corollary}
\newtheorem{Lemma}{Lemma}
\newtheorem{Proposition}{Proposition}
\newtheorem*{Definition}{Definition}

\title{\Huge{Ground states of \\self-gravitating elastic bodies}}
\author{Simone Calogero\footnote{E-Mail: calogero@ugr.es}\\[0.2cm]
Departamento de Matem\'atica Aplicada\\
Facultad de Ciencias, Universidad de Granada\\
Campus de Fuentenueva, 18071 Granada, Spain\\[0.5cm]
Tommaso Leonori\thanks{E-Mail:  tleonori@math.uc3m.es}\\[0.2cm]
Departamento de Matem\'aticas\\ Universidad Carlos III de Madrid\\Avenida de la Universidad 30,
28911 Legan\'es, 
Spain}
\date{}
\begin{document}

\maketitle
\abstract{The existence of static, self-gravitating elastic bodies in the non-linear theory of elasticity is established. Equilibrium configurations of self-gravitating elastic bodies close to the reference configuration have been constructed in~\cite{BS} using the implicit function theorem. In contrast, the steady states considered in this article correspond to deformations of the relaxed state with no size restriction and are obtained as minimizers of the energy functional of the elastic body.}

\section{Introduction}

Steady states of bodies subject to their own self-induced gravitational field are at the center of interest in theoretical astrophysics. They describe different physical systems according to the specific type of material making up the body. The most popular examples are steady states of self-gravitating fluids, which describe equilibrium configurations of stars~\cite{CH}, and of kinetic collisionless matter (Vlasov matter), which are models for steady galaxies~\cite{BT}. For these important examples there is a vast mathematical literature 
concerning the existence theory of static solutions, see for instance~\cite{CSS,JL,LMR,GR2,GR,GR3,W}. 
In contrast, the existence of static self-gravitating elastic bodies has been scarcely investigated so far, despite the prominent role that these models play in the context of neutron stars physics~\cite{Ch}. In fact, the only result that we are aware of is~\cite{BS} (and its extension to the general relativistic case~\cite{ABS}), where R.~Beig and B.~G.~Schmidt prove the existence of static, self-gravitating elastic bodies near the reference configuration by using the implicit function theorem. In the present article we consider elastic bodies that correspond to deformations of the reference state without any {\it a priori} size restriction. 
The precise formulation of our results is given in Section~\ref{mainres}. 
This Introduction continues with a general discussion on the problem of self-gravitating bodies in equilibrium. 

In the approximation in which thermal effects can be neglected, the mass density $\rho(x)$ of isolated, self-gravitating matter in equilibrium satisfies the equation
\begin{equation}\label{generaleq}
\mathrm{Div}\,\sigma=\rho\nabla V,\quad V(x)=-\int_{\R^3}\frac{\rho(y)}{|x-y|}\,dy,\quad x\in\R^3,
\end{equation}
where $\sigma=\sigma(x)$ is the Cauchy stress tensor and $V=V(x)$ is the gravitational potential self-induced by the matter distribution. 

The system~\eqref{generaleq} must be complemented by a constitutive law expressing the stress tensor in terms of the mass density. For simplicity we assume that this is given by an explicit equation $\sigma=\hat{\sigma}(x,\rho)$, with $\hat{\sigma}_{ij}(x,0)=0$, for all $x\in\R^3$ and $i,j=1,2,3$. For instance, for a barotropic fluid with pressure $p=p(\rho)$, $p(0)=0$, the stress tensor takes the form $\sigma_{ij}=-p\, \delta_{ij}$,  where 
$\delta_{ij}=1 $ if $j=i$ and $\delta_{ij}=0 $ otherwise.

\begin{Definition}[{\bf Single Body}]
The matter distribution is called a single body if there exists a non-empty open bounded connected set $\Omega\subset\R^3$ with boundary $\partial\Omega$ of zero Lebesgue measure such that $\rho>0$, for $x\in\Omega$ and $\rho=0$, for $\overline{\Omega}^{\,c}$. 
\end{Definition}
Note that our definition of single body does {\it not} require regularity of the boundary $\partial\Omega$. Moreover,
since the mass density of a single body need not vanish on $\partial\Omega$,  then $\rho$ is in general only a weak solution of~\eqref{generaleq}. We recall that a function $\rho:\R^3\to [0,\infty)$ is a weak solution of~\eqref{generaleq} if  $\rho\in L^{1} \cap L^{3/2}(\R^3)$, $\sigma\in L^1(\R^3)$ and the following integral equation holds: 
\begin{subequations}\label{generaleqweak}
\begin{align}
&\int_{\R^3}\tr(\sigma\cdot\nabla\phi^{\,T})\,dx=-\int_{\R^3}\int_{\R^3}\frac{\rho(x)\rho(y)}{|x-y|^3}\,(x-y)\cdot\phi(x)\,dx\,dy,\\[0.5cm]
&\text{for all $\phi:\R^3\to\R^3$, $\phi\in C^1_c(\R^3)$}, 
\end{align}
\end{subequations}
where  $A\cdot B$ denotes the dot product of the tensors $A$ and $B$ (i.e., the contraction of the last index of $A$ with the first index of $B$), while $A^{\,T}$ and $\tr(A)$ denote the transpose and the trace of the second order tensor $A$, respectively  (our convention for the gradient is $(\nabla f)_{ij}=\partial_{x^j}f_i$, for all functions $f:\R^3\to\R^3$).

For single body solutions, we may restrict the integrals in~\eqref{generaleqweak} over the support $\Omega$, and obtain that $\rho$ is a weak solution of~\eqref{generaleq} if the the mass density in the interior of the body solves the integral equation 
\begin{subequations}\label{generaleqbounded}
\begin{align}
&\int_{\Omega}\tr(\sigma\cdot\nabla\eta^{\,T})\,dx=-\int_{\Omega}\int_{\Omega}\frac{\rho(x)\rho(y)}{|x-y|^3}\,(x-y)\cdot\eta(x)\,dx\,dy,\label{equation2}\\[0.5cm]
&\text{for all $\eta:\Omega\to\R^3$, $\eta\in \mathcal{D}_{\Omega}(\R^3)$,}
\end{align}
where $\mathcal{D}_{\Omega}(\R^3)$ is the trace 
of $C_{c}^1(\R^3)$ on $\Omega$, i.e., the space
\begin{equation}
\mathcal{D}_{\Omega}(\R^3)=\{\eta \in C^1(\Omega) \, :\, \eta=\phi_{|_{\Omega}},\ \text{ for some }\phi:\R^3\to\R^3,\ \phi\in C^1_c(\R^3)\}.
\end{equation}
\end{subequations}
Since $\mathcal{D}_\Omega(\R^3)\subset C^1(\overline{\Omega})\subset H^1(\Omega)$, a subset of solutions are those that solve~\eqref{equation2} for all $\eta\in H^1(\Omega)$. In the latter case, and assuming that $\partial\Omega$ is smooth, $\rho$ is a weak solution of the following Dirichlet-type  boundary value problem
\begin{subequations}\label{boundaryvalueproblem}
 \begin{align}
&\mathrm{Div}\,\sigma=\rho\displaystyle{\int_{\Omega}\frac{\rho(y)}{|x-y|^2}\frac{x-y}{|x-y|}\,dy},& x\in\Omega,\label{equation}\\[0.5cm]
&\rho=\rho_b:\sigma(x,\rho_b)\cdot n=0,&x\in\partial\Omega\label{boundary},
\end{align}
\end{subequations}
where $n$ denotes the exterior unit normal vector field on $\partial\Omega$. We also remark that compactly supported solutions of~\eqref{generaleq} might exist only for special domains $\Omega$. For instance, isolated, self-gravitating fluid bodies are necessarily spherically symmetric~\cite{Li, Lin} and therefore they are all supported on a ball.


 
Problem~\eqref{generaleqbounded} is posed in the Euler formulation of continuum mechanics, which is the natural one for fluids. However for self-gravitating elastic bodies the problem is more naturally formulated in the Lagrangian picture, since the stress tensor of an elastic body depends on the deformation of the body; see~\cite{C} for an introduction to 3-dimensional non-linear elasticity.

Let $\mathcal{B}$ be an open, bounded, connected subset of $\R^3$ with smooth boundary and $\psi:\overline{\mathcal{B}}\to\R^3$ be an injective function on $\mathcal{B}$.  Denote
\[
F=\nabla\psi,\quad  \mathrm{Cof}\,F= F^{-T}\det F\ \text{(the matrix of cofactors of $F$)}.
\]
It is assumed that $\psi$ preserves orientation, i.e., $\det\nabla\psi (X) >0$, for all $X\in \overline{\mathcal{B}}$. The set $\overline{\mathcal{B}}$ is called reference configuration (or material manifold) and identifies the body  in a completely relaxed state, before any deformation takes place. Upon a deformation, the ``particle" $X$ of the body is moved to the new position $x=\psi(X)$. If we denote by $\rho_\mathrm{ref}:\overline{\mathcal{B}}\to (0,\infty)$ the mass density of the body in the reference configuration, and by $\rho:\psi(\overline{\mathcal{B}})\to (0,\infty)$ the mass density in the deformed state, the conservation of mass entails
\begin{equation}\label{local}
\rho(\psi(X))=\frac{\rho_\mathrm{ref}(X)}{\det\nabla\psi(X)}.
\end{equation}
The (first) Piola-Kirchhoff stress tensor is defined by
\begin{equation}\label{piola}
\Sigma(X)=\sigma (X) \cdot \mathrm{Cof}\,F(\psi(X)).
\end{equation}
The change of variable $x=\psi(X)$, together with~\eqref{local} and~\eqref{piola}  turns
 equation~\eqref{equation} into the form
\begin{equation}\label{generaleqL}
\mathrm{Div}\,\Sigma=\rho_\mathrm{ref}\int_\mathcal{B}\rho_\mathrm{ref}(Y)\frac{\psi(X)-\psi(Y)}{|\psi(X)-\psi(Y)|^3}\,dY, \quad X\in \mathcal{B}.
\end{equation}  
By inverting~\eqref{local} and~\eqref{piola}, sufficiently smooth solutions of~\eqref{generaleqL} are transformed into solutions of~\eqref{equation} in the domain $\Omega=\psi(\mathcal{B})$. If the deformation $\psi$ is injective on the closure of $\mathcal{B}$, then $\Omega$ is a domain with smooth boundary. However, if injectivity is lost at the boundary, as it is often the case, there is no guarantee that the set $\Omega$ possesses a smooth boundary, which is the main reason for introducing the definition of single body as above. We remark that interior injectivity prevents the interpenetration of matter, while injectivity up to the boundary prevents that even self-contact of the boundary of the body occurred.


The body is said to be elastic if $\Sigma$ depends on the configuration $\psi$ only through the deformation gradient $F=\nabla\psi$. Denoting by $\mathbb{M}^3$ the linear space of $3\times 3$ real matrices, and by $\mathbb{M}^3_+$ the subset thereof consisting of matrices with positive determinant, we then assume the existence of a function $\hat{\Sigma}:\mathcal{B}\times\mathbb{M}^3_+\to\mathbb{M}^3$ such that $\Sigma(X)=\hat{\Sigma}(X,F(X))$.
For elastic bodies we may introduce a variational formulation of~\eqref{generaleqL} as follows. Let a function $w:\mathcal{B}\times\mathbb{M}^3_+\to [0,\infty)$ be given and let us define the energy functional of the elastic body as
\begin{equation}\label{energyelastic}
I[\psi]=\int_\mathcal{B} \rho_\mathrm{ref}(X)w(X,\nabla\psi)\,dX-\frac{1}{2}\int_\mathcal{B}\int_\mathcal{B}\frac{\rho_\mathrm{ref}(X)\rho_\mathrm{ref}(Y)}{|\psi(X)-\psi(Y)|}\,dX\,dY.
\end{equation}
The function $w$ is the stored energy of the elastic body.  
The first term in the energy functional is the total strain energy, while the second term is the gravitational potential energy of the body. A simple formal calculation shows that~\eqref{generaleqL} corresponds to the Euler-Lagrange equation associated to the functional~\eqref{energyelastic} for a Piola-Kirchhoff stress tensor given by
 \begin{equation}\label{piolakir}
\Sigma=\rho_\mathrm{ref}\frac{\partial w}{\partial F}.
\end{equation}
This suggests that, under appropriate assumptions on the growth and the regularity of $w$,
one can obtain solutions of~\eqref{generaleqL}  by proving the existence of minimizers of the functional $I$ over a suitable space. The choice of the minimizing space is one of the main issues of the problem. We shall consider two different minimizing spaces, $\mathscr{A}^{(1)}$ and $\mathscr{A}^{(2)}$, which in the pure elastic case were first considered in~\cite{ball2,CN}. The deformations in the space $\mathscr{A}^{(1)}$ preserve the mass of any local subregion of the body, as well as the center of mass  in the deformed state, but they are not in general injective functions, not even in the interior of the body. On the other hand, the deformations in the space $\mathscr{A}^{(2)}$ are homeomorphisms on $\overline{\mathcal{B}}$ that preserve the shape of the body in the deformed state.

The results concerning the existence of a  minimum of the functional~\eqref{energyelastic} are contained in Theorem \ref{maintheo} below. 
Once the existence of minimizers of the functional~\eqref{energyelastic} has been proved, it is natural to ask whether they are solutions of~\eqref{generaleqL}. We are not able to give a positive answer to this question; even in the absence of self-gravitating interaction this is a major open problem in elasticity theory~\cite{balla}. However, following the ideas introduced by J.~Ball in~\cite{ball}, we are able to show that 
\begin{enumerate}
\item The spatial density $\rho^{(1)}$ associated to the minimizer $\psi^{(1)}$ in the space $\mathscr{A}^{(1)}$ is a solution of~\eqref{generaleqbounded} in the region $\Omega^{(1)}=\psi^{(1)}(\mathcal{B})$, and thus describes the interior mass density of an isolated, self-gravitating body in equilibrium; the boundary of $\Omega^{(1)}$ has zero Lebesgue measure, but we cannot say anything about its regularity. 
\item The spatial density $\rho^{(2)}$ associated to the minimizer $\psi^{(2)}$ in the space $\mathscr{A}^{(2)}$  solves equation~\eqref{equation} in a distributional sense in the domain $\Omega^{(2)}=\psi^{(2)}(\mathcal{B})$; in this case, the boundary of $\Omega^{(2)}$ is regular, but since $\rho^{(2)}$ need not satisfy~\eqref{generaleqweak}, the body will not in general be isolated. 
\end{enumerate} 
The precise statement of our main results is given in Section~\ref{mainres}; their proofs are to be found in Section~\ref{existence}.

\section{Preliminaries and main results}\label{mainres} 

We assume that $\mathcal{B}$ is a {\it regular domain}, i.e., a non-empty, open, bounded, connected subset of $\R^3$ with Lipschitz continuous boundary; in particular, $\mathrm{meas}\,\partial\mathcal{B}=0$. The mass density in the reference configuration $\rho_\mathrm{ref}:\overline{\mathcal{B}}\to (0,\infty)$ is assumed to satisfy
\begin{equation}\label{rhoref}
\rho_\mathrm{ref}\in  L^\infty (\mathcal{B}),\quad\mbox{ess}\inf_{\mathcal{B}}\, \rho_\mathrm{ref}(X)= \rho_0>0
\end{equation}
and we denote 
\[
M=\|\rho_\mathrm{ref}\|_{L^1}.
\] 
We write the energy functional of the elastic body as
\[
I[\psi]=E_\mathrm{str}[\psi]+E_\mathrm{pot}[\psi],
\] 
where $E_\mathrm{pot}[\psi]$ is the potential energy, 
\begin{subequations}
\begin{align}\label{potentialen}
&E_\mathrm{pot}[\psi]=-\frac{1}{2}\int_\mathcal{B}\int_\mathcal{B}\Theta_\psi(X,Y)\rho_\mathrm{ref}(X)\rho_\mathrm{ref}(Y)\,dX\,dY,\\
&\Theta_\psi(X,Y)=\frac{1}{|\psi(X)-\psi(Y)|},
\end{align}
\end{subequations}
while $E_\mathrm{str}[\psi]$ is the total strain energy of the body,
\begin{equation}\label{strainen}
E_\mathrm{str}[\psi]=\int_\mathcal{B} \rho_\mathrm{ref}(X)w(X,\nabla\psi)\,dX.
\end{equation}
Our first goal is to prove the existence of a minimizer to the functional $I$ in two different spaces, which we denote $\mathscr{A}^{(1)}$ and $\mathscr{A}^{(2)}$.
We assume $\mathscr{A}^{(i)}\subset W^{1,\beta}(\mathcal{B})$, for $\beta>3$, and henceforth any function in the Sobolev space $W^{1,\beta}(\mathcal{B})$ will be identified with its representative in $C^0(\overline{\mathcal{B}})$. Furthermore we assume that the elements of $\mathscr{A}^{(i)}$ satisfy $\det\nabla\psi>0$ for almost all $X\in \mathcal{B}$. Since $\nabla\psi=0$ a.e. on any set in which $\psi$ is constant (see for instance~\cite[Lemma~7.7]{GT}), we infer that $|\psi(X)-\psi(Y)|>0$ for almost all $X,Y\in\mathcal{B}\times\mathcal{B}$ and therefore the function $\Theta(X,Y)$ that appears in the potential energy~\eqref{potentialen} is well-defined almost everywhere in $\mathcal{B}\times\mathcal{B}$.

Before completing the definition of the spaces $\mathscr{A}^{(i)}$, we need to list our assumptions on the stored energy function: 
\begin{itemize}
\item[(w1)] Polyconvexity: There exists a convex function $\hat{w}(X,\cdot):\mathbb{M}^3\times\mathbb{M}^3\times (0,+\infty)\to\R$ such that $w(X,F)=\hat{w}(X,F,\mathrm{Cof}\,F,\det F)$, for all $F\in\mathbb{M}^3_+$ and for almost all $X\in \mathcal{B}$;
\item[(w2)] Regularity: The function $\hat{w}(\cdot,F,H,\delta)$ is measurable for all $(F,H,\delta)\in\mathbb{M}^3\times\mathbb{M}^3\times (0,\infty)$.
\end{itemize}
In addition we impose a lower bound on the stored energy function which depends on the minimizing space considered. For $F\in\mathbb{M}^3$ we denote $|F|=\sqrt{\mathrm{Tr}(F\cdot F^T)}$ (or any other equivalent matrix norm), while $|a|$ stands for the standard Euclidean norm when $a$ is a real number. When working in the space $\mathscr{A}^{(1)}$ we assume that the stored energy functions satisfies
\begin{itemize}
\item[(w3)$^{(1)}$] There exist
\begin{equation}\label{esse}
p>6,\quad q\geq \frac{p}{p-1},\quad s>\frac{2p}{p-6},
\end{equation}
a constant $\alpha>0$ and a function $h\in L^1 (\mathcal{B})$ such that
\[
w(X,F)\geq \alpha(|\!\det F|^{-s}+|F|^p+|\mathrm{Cof}\,F|^q)+ h(X), 
\]
for almost all $X\in \mathcal{B}$.
\end{itemize}
When working in the space $\mathscr{A}^{(2)}$ we assume  
\begin{itemize}
\item[(w3)$^{(2)}$] There exist
\begin{equation}\label{esse2}
p>3,\quad q>3,\quad s>\frac{2q}{q-3},
\end{equation}
a constant $\alpha>0$ and a function $h\in L^1 (\mathcal{B})$ such that
\[
w(X,F)\geq \alpha(|\!\det F|^{-s}+|F|^p+|\mathrm{Cof}\,F|^q)+ h(X), 
\]
for almost all $X\in \mathcal{B}$.
\end{itemize}
Up to the choice of the exponents $p,q,s$, the assumptions listed above are standard in elasticity theory. For a throughout discussion on stored energy functions and examples that satisfy (or do not) the above conditions, we refer to~\cite{Ball} and~\cite[Chs.~3-4]{C}. At the end of the present section we shall discuss briefly the class of Ogden materials. We also remark that there exist additional physical conditions that the stored energy  functions must satisfy, which however are irrelevant for our analysis. For example, $w$ must be normalized such that
\begin{equation}\label{normalization}
w(X,\mathbb{I})=0,\quad\text{ for all $X\in\mathcal{B}$}.
\end{equation} 

Next we complete the definition and prove some basic properties of the spaces $\mathscr{A}^{(1)}$, $\mathscr{A}^{(2)}$. 

\subsubsection*{The space $\mathscr{A}^{(1)}$}
Let
\begin{equation}\label{setS}
S_\psi=\{x\in\psi(B) : \mathrm{card}\,\{\psi^{-1}(x)\}>1\},
\end{equation}
where $\{\psi^{-1}(x)\}$ denotes the pre-image set of $x\in\psi(B)$ and $\mathrm{card}\,U$ denotes the cardinality of the set $U$, and let $a\in\R^3$. Given a stored energy function that satisfies the properties (w1), (w2), (w3)$^{(1)}$, we define the space $\mathscr{A}^{(1)}$ as 
\begin{align*}
\mathscr{A}^{(1)}=\{\psi\in W^{1,p}(\mathcal{B}):\ &
E_\mathrm{str}[\psi]<\infty,\\ 
& \det\nabla\psi>0\, \text{ a.e. on $\mathcal{B}$,}\\  
&\mathrm{meas}\,S_\psi=0,\\
&\int_\mathcal{B}\rho_\mathrm{ref}(X)\psi(X)\,dX=a\}.
\end{align*}
Since we can always find a constant matrix $C\in\mathbb{M}^3_+$ and a vector $d\in\R^3$ (depending on $a$ and $\rho_\mathrm{ref}$) such that $\psi(X)=C\cdot X+d\in\mathscr{A}^{(1)}$, then $\mathscr{A}^{(1)}$ is not empty. 

A function $\psi$ that satisfies the condition $\mathrm{meas}\,S_\psi=0$ is often referred to as a.e. injective~\cite{C}.
The meaning of the last condition in the definition of $\mathscr{A}^{(1)}$ will be clarified shortly. 
Before proving some basic properties of functions in the space $\mathscr{A}^{(1)}$, let us note that
\[
S_\psi=\psi(K_\psi),\quad \mbox{where} \quad K_\psi=\{X\in\mathcal{B}:\exists\, Y\in\mathcal{B}, Y\neq X, \text{ such that } \psi(X)=\psi(Y)\}
\]
and we recall that the change of variables formula for Sobolev maps $\psi\in W^{1,\beta}(\mathcal{B})$, $\beta>3$, is~\cite[Th.~2]{MM}
\begin{equation}\label{changevariablegen}
\int_Uf(\psi(X))\,\det\nabla\psi(X)\,dX=\int_{\psi(U)}f(x)\,\mathrm{card}\,\{\psi^{-1}(x)\}\,dx,
\end{equation}
for all measurable sets $U\subset\mathcal{B}$ and measurable functions $f:\psi(U)\to\R$. In particular, since functions in the space $\mathscr{A}^{(1)}$ are a.e.~injective we obtain
\begin{equation}\label{changevariable}
\int_Uf(\psi(X))\,\det\nabla\psi(X)\,dX=\int_{\psi(U)}f(x)\,dx,\quad\text{for all $\psi\in\mathscr{A}^{(1)}$}.
\end{equation}
\begin{Lemma}\label{propertiesA1}
For all $\psi\in\mathscr{A}^{(1)}$ the following holds:
\begin{itemize}
\item[\rm{(i)}] $\psi$ maps null sets into null sets;
\item[\rm{(ii)}] $\psi$ is open;
\item[\rm{(iii)}] $\psi(\mathcal{B})$ is open, connected and satisfies
\begin{equation}\label{inclusions}
\overline{\psi(\mathcal{B})}=\psi(\overline{\mathcal{B}}),\quad \partial \psi(\mathcal{B})\subset\psi(\partial\mathcal{B});
\end{equation}
\item[\rm{(iv)}] $\mathrm{meas}\, K_\psi=\mathrm{meas}\,\partial\psi(\mathcal{B})=0$;
\item[\rm{(v)}] The spatial density
\begin{equation}\label{spatialdensity}
\rho_\psi(x)=\frac{\rho_\mathrm{ref}(\psi^{-1}(x))}{\det\nabla\psi(\psi^{-1}(x))},
\end{equation}
is almost everywhere defined and positive on $\psi(\mathcal{B})$, and satisfies
\begin{equation}\label{localmasscons}
\int_U\rho_\mathrm{ref}(X)\,dX=\int_{\psi(U)}\rho_\psi(x)\,dx,\quad\text{for all $\psi\in\mathscr{A}^{(1)}$}. 
\end{equation}
In particular, $\|\rho\|_{L^1(\R^3)}=M$, for all $\psi\in\mathscr{A}^{(1)}$.
\item[\rm{(vi)}] The last condition in the definition of the space $\mathcal{A}^{(1)}$  can be rewritten as
\begin{equation}\label{com}
\int_{\psi(\mathcal{B})}x\,\rho_\psi(x)\,dx=a,\quad \text{for all $\psi\in\mathscr{A}^{(1)}$;}
\end{equation}
\item[\rm{(vii)}]  The potential energy can be expressed as
\begin{subequations}\label{epot}
\begin{equation}
E_\mathrm{pot}[\psi]=e_\mathrm{pot}[\overline{\rho}_\psi],
\end{equation}
where, denoting by $\mathbb{I}_U$ the characteristic function of the set $U$, 
\begin{equation}
\overline{\rho}_\psi=\rho_\psi\mathbb{I}_{\psi(\mathcal{B})},\quad e_\mathrm{pot}[\rho]=\int_{\R^3}\int_{\R^3}\frac{\rho(x)\rho(y)}{|x-y|}\,dx\,dy.
\end{equation}
\end{subequations}
\end{itemize}
\end{Lemma}
\begin{proof}
The property (i) is valid in general for all functions $\psi\in W^{1,\beta}(\mathcal{B})$, $\beta>3$, see~\cite{MM}.
To prove (ii), we first note that 
\[
\int_{\mathcal{B}}\left(|\nabla\psi|^p+(\det\nabla\psi)^{-s}\right)dX<\infty,
\]
for all $\psi\in\mathscr{A}^{(1)}$, where $p,s$ are given by~\eqref{esse}. 
Let us introduce the {\it outer distortion} $K_O$ of $\psi$:
\[
K_O(X)=\frac{|\nabla\psi|^3}{\det\nabla\psi}.
\]
By H\"older's inequality, $K_O\in L^\beta(\mathcal{B})$, $\beta=ps/(p+3s)$. Since $\beta>2$, then (ii) follows by~\cite[Cor.~1.10]{HK}. Next we prove (iii). By (ii), $\psi(\mathcal{B})$ is open; since continuous functions map connected sets into connected sets, it is also connected. The properties~\eqref{inclusions} are proved in~\cite[Thm.~1.2-7]{C} under the hypothesis that $\psi$ is injective on $\mathcal{B}$, but the same result, with the same simple proof, holds under the (weaker) assumption that $\psi$ is open.  
As to (iv), applying the change of variable formula~\eqref{changevariable} we have
\[
\int_{K_\psi}\det\nabla\psi\,dX=\int_{S_\psi}dx=\mathrm{meas}\,S_{\psi}=0.
\]  
Since $\det\nabla\psi>0$ a.e. in $\mathcal{B}$, then meas$\,K_\psi$=0. Moreover, meas$\,\partial\psi(\mathcal{B})=0$ follows by (i) and the inclusion $\partial \psi(\mathcal{B})\subset\psi(\partial\mathcal{B})$. To prove (v), let
\[
T_\psi=\{x\in\psi(B):\det\nabla\psi(X)=0,\ \text{for some }X\in\{\psi^{-1}(x)\}\}.
\]
Since $T_\psi=\psi(\{X\in B: \det\nabla\psi(X)=0\})$, then meas\,$T_\psi=0$. As the spatial density~\eqref{spatialdensity} is well-defined and positive for $x\in\psi(B)\setminus(S_\psi\cup T_\psi)$, the first part of the claim follows. The identities~\eqref{localmasscons},~\eqref{com} and~\eqref{epot} are proved using~\eqref{changevariable}. 
\end{proof}
We remark that by~\eqref{localmasscons}, 
 {\it the deformations $\psi\in\mathscr{A}^{(1)}$ preserve the mass of any subregion of the body}, while by~\eqref{com},  {\it the deformations $\psi\in\mathscr{A}^{(1)}$ leave invariant the center of mass of the body in the deformed state.}

\subsubsection*{The space $\mathscr{A}^{(2)}$}
For the definition of $\mathcal{A}^{(2)}$, 
let a function $\zeta:\overline{\mathcal{B}}\to\R^3$, $\zeta\in W^{1,p}(\mathcal{B})$, $p>3$, be given with the following properties:
\begin{itemize}
\item[(z1)] $\zeta$ is one-to-one on $\overline{\mathcal{B}}$, $\det\nabla\zeta>0$ a.e. in $\mathcal{B}$ and $E_\mathrm{str}[\zeta]<\infty$;
\item[(z2)] The boundary of the set $\zeta(\mathcal{B})$ is Lipschitz continuous.
\end{itemize}
Note that under assumptions (z1)-(z2), $\zeta(\mathcal{B})$ is a regular domain and that a sufficient condition for (z2) is that $\zeta$ be a $C^1$ diffeomorphism on $\overline{\mathcal{B}}$. Moreover by~\cite[Th~1.2-8]{C} we have
\begin{equation}\label{domainzeta}
\zeta(\overline{\mathcal{B}})=\overline{\zeta(\mathcal{B})},\quad \partial\zeta(\mathcal{B})=\zeta(\partial\mathcal{B}).
\end{equation}  
Given a stored energy function that satisfies the properties (w1), (w2), (w3)$^{(2)}$ and a function $\zeta$ as above, we define the minimizing space $\mathscr{A}^{(2)}$ as
\begin{align*}
\mathscr{A}^{(2)}=\{\psi\in W^{1,p}(\mathcal{B}):\ &
E_\mathrm{str}[\psi]<\infty,\\ 
& \det\nabla\psi>0\,\text{ a.e. on $\mathcal{B}$,}\\
& \psi=\zeta\text{ for }x\in\partial\mathcal{B}\}.
\end{align*}
Since $\zeta\in\mathscr{A}^{(2)}$, then $\mathscr{A}^{(2)}$ is not empty.
\begin{Lemma}\label{propertiesA2}
The properties listed in Lemma~\ref{propertiesA1} are also valid for all $\psi\in\mathscr{A}^{(2)}$, and in addition the following holds:
\begin{itemize} 
\item[\rm{(A)}] $\psi(\overline{\mathcal{B}})=\zeta(\overline{\mathcal{B}})$;
\item[\rm{(B)}] $\psi$ is a homeomorphism of $\overline{\mathcal{B}}$ onto $\zeta(\overline{\mathcal{B}})$; 
\item[\rm{(C)}] The inverse function $\psi^{-1}$ belongs to $W^{1,r}(\zeta(\mathcal{B}))$, where
\begin{equation}\label{expr}
r=\frac{q(1+s)}{q+s}\quad(r>3);
\end{equation}
\item[\rm{(D)}] There holds
\begin{equation}\label{inclusion2}
\psi(\mathcal{B})=\zeta(\mathcal{B}),\quad \partial\psi(\mathcal{B})=\partial\zeta (\mathcal{B}).
\end{equation}
In particular, $\psi(\mathcal{B})$ is a regular domain.
\end{itemize}
\end{Lemma}
\begin{proof}
Since
\[
\int_\mathcal{B}(|\mathrm{cof}\nabla\psi|^q+(\det\nabla\psi)^{-s})\,dX<\infty,\quad\text{for $\psi\in\mathscr{A}^{(2)}$},
\]
with $q,s$ given by~\eqref{esse2}, H\"older's inequality entails
\[
\int_\mathcal{B}|(\nabla\psi)^{-1}|^r\det\nabla\psi\,dX=\int_{\mathcal{B}}(\det\nabla\psi)^{1-r}|\mathrm{Cof}\nabla\psi|^r\,dX<\infty,
\]
where $r$ is given by~\eqref{expr}. Thus the claims (A), (B) and (C) follow by~\cite[Th.~2]{ball2}.
Applying again~\cite[Th~1.2-8]{C} we infer that the properties~\eqref{domainzeta} are also satisfied by $\psi$, which gives (D). Finally, all the properties listed in Lemma~\ref{propertiesA1} also hold for $\psi\in\mathscr{A}^{(2)}$, because injective continuous functions are open (Invariance of Domain Theorem) and satisfy $S_\psi=\emptyset$.
\end{proof}
By~\eqref{inclusion2}, it follows that {\it the deformations in the space $\mathscr{A}^{(2)}$ preserve the shape of the body in the deformed state} and that {\it self-contact of the boundary does not occur}.


\subsubsection*{Main results}
Our first  main result is the following:
\begin{Theorem}\label{maintheo} Let $\rho_\mathrm{ref}$ satisfy~\eqref{rhoref}. Fix $i=1$ or $2$. Let $w$ satisfy (w1), (w2), (w3)$^{(i)}$, and, for $i=2$, let $\zeta$ satisfy (z1), (z2). Then there exists $\psi^{(i)}\in\mathscr{A}^{(i)}$ such that 
\[
I[\psi^{(i)}]=\inf_{\mathscr{A}^{(i)}}I[\psi]. 
\]
Moreover, the spatial density 
\[
\rho^{(i)}=\rho_{\psi^{(i)}}
\]
satisfies $\rho^{(i)}\in L^\gamma(\psi^{(i)}(\mathcal{B}))$, for all $1\leq\gamma\leq 1+s$.
\end{Theorem}

Next we study the relation between the minimizers of the energy functional and the mass density of static, self-gravitating bodies.  To this purpose we need an additional assumption on the stored energy function:
\begin{itemize}
\item[(w4)] $w(X,\cdot)$ is $C^1$ and there exists a constant $K>0$ such that
\[
\left|\frac{\partial w}{\partial F}\cdot F^T\right|\leq K(w(X,F)+1).
\]
\end{itemize}
Let us denote 
\[
\Omega^{(i)}=\psi^{(i)}(\mathcal{B}).
\]
By Lemma~\ref{propertiesA1}(iii) and Lemma~\ref{propertiesA2}(D), $\Omega^{(i)}$ is open, connected and $\mathrm{meas}\,\partial\Omega^{(i)}=0$; $\Omega^{(2)}$, in addition, has a Lipschitz continuous boundary and so is a regular domain. Define the Cauchy stress tensor $\sigma^{(i)}:\Omega^{(i)}\to\mathbb{M}^3$ as
\begin{equation}\label{cauchystress}
\sigma^{(i)}(x)=\left(\frac{\Sigma(X,\nabla\psi^{(i)}(X))\cdot\nabla\psi^{(i)}(X)^T}{\det\nabla\psi^{(i)}(X)}\right)_{|_{X=(\psi^{(i)})^{-1}(x)}},\quad\Sigma=\rho_\mathrm{ref}\frac{\partial w}{\partial F}.
\end{equation}

Our result concerning the equation solved by the minimizers is the following.

\begin{Theorem}\label{propmin}
In addition to the hypotheses of Theorem~\ref{maintheo}, let $w$ satisfy (w4). 
Then  $\sigma^{(i)}\in L^1(\Omega^{(i)})$ and for $i=1$ the following identity holds:
\begin{equation}\label{viral}
e_\mathrm{pot}[\overline{\rho}_{(1)}]=\int_{\R^3}\tr\,\overline{\sigma}_{(1)}\,dx,
\end{equation}
where
\[
\overline{\rho}_{(1)}=\rho^{(1)}\mathbb{I}_{\Omega^{(1)}},\quad \overline{\sigma}_{(1)}=\sigma^{(1)}\mathbb{I}_{\Omega^{(1)}}.
\]
Moreover $(\rho^{(1)},\sigma^{(1)})$ solves~\eqref{generaleqbounded}, while $(\rho^{(2)},\sigma^{(2)})$ satisfies~\eqref{equation2} for all functions $\eta\in C^{1}(\overline{\Omega^{(2)}})$ such that $\eta=0$ for $x\in\partial\Omega^{(2)}$.
\end{Theorem}
Since the pair $(\overline{\rho}_{(1)},\overline{\sigma}_{(1)})$ solves~\eqref{generaleqweak},  then  $\overline{\rho}_{(1)}$ is the mass density of an isolated single body supported in the region $\Omega^{(1)}$. Note however that we prove no results neither on the regularity of the boundary, nor on the regularity of the mass density in the interior of the body. These important open problems require further investigation.
 
The identity~\eqref{viral} is an example of ``Virial Theorem", which in the context of particle mechanics gives information on how the energy of a system in equilibrium is distributed between kinetic and potential energy. Note finally that in view of~\eqref{normalization}, the energy of the minimizer $\psi^{(1)}$ is negative, since $\psi(X)=X+d$ belongs to $\mathscr{A}^{(1)}$ for 
\[
d=\frac{1}{M}\left(a-\int_\mathcal{B}\rho_\mathrm{ref}(X)X\,dX\right).
\]


\subsubsection*{Examples of stored energy functions}\label{examples}
The stored energy function of Ogden materials is given by
\begin{equation}\label{ogden}
w(X,F)=\sum_{i=1}^La_i(X)(\tr C)^{\gamma_i/2}+\sum_{j=1}^Nb_j(X)(\tr\mathrm{Cof}\,C)^{\delta_j/2}+\Gamma(\det F)+h(X),
\end{equation}
where $L,N\in\mathbb{N}$, $C=F^T\cdot F$ is the (left) Cauchy-Green tensor, $\gamma_i\geq 1$, for $i=1,\dots L$, $\delta_j\geq 1$, for $j=1,\dots N$, $\Gamma:(0,\infty)\to (0,\infty)$ is a $C^1$ convex function such that 
\[
\Gamma(z)\geq c_1 z^{-s},\quad\text{for some $c_1>0$},
\]
$a_i,b_j:\mathcal{B}\to (0,\infty)$ are measurable bounded functions such that
\begin{align*}
&\bar{a}=\mbox{ess\,inf}_{X\in\mathcal{B}}\min_{i=1,\dots L} a_i(X)>0,\\
&\bar{b}=\mbox{ess\,inf}_{X\in\mathcal{B}}\min_{j=1,\dots N} b_j(X)>0.
\end{align*}
The function $h$ is defined in such a way that the normalization condition~\eqref{normalization} holds.
As shown in~\cite[Th.~4.9-2]{C}, the stored energy function~\eqref{ogden} is polyconvex and satisfies the coercive inequality in (w3)$^{(i)}$ for 
\begin{equation}\label{exppq}
p=\max_{i}\gamma_i,\quad q=\max_j\delta_j,
\end{equation}
and a constant $\alpha>0$ depending on $\bar{a},\bar{b}, c_1$. 
Moreover, assuming that
\[
|\Gamma'(z)|\leq \frac{c_2}{z}(1+\Gamma(z)),\quad\text{for some $c_2>0$},
\]
the stored energy function~\eqref{ogden} satisfies (w4) as well, see~\cite[Sec.~2.4]{ball}. We remark that Ogden materials also meet other desirable physical requirements, such as {\it material frame indifference}, and provide a rich supply of case studies with important practical and theoretical applications~\cite{OSS}.

\section{Proof of the main results}\label{existence}

\subsection{Proof of Theorem~\ref{maintheo}}
To begin with we recall a result which will be used to show that the sets $\mathscr{A}^{(1)}$ and $\mathscr{A}^{(2)}$ are weakly closed. The proof can be found in~\cite{Ball}, see also~\cite[Thm.~7.6-1]{C}. 
\begin{Lemma}\label{tools}
Let $\{\psi_n\}\subset W^{1,a}(\mathcal{B})$ such that $\psi_n\rightharpoonup\psi$ in $W^{1,a}(\mathcal{B})$, $\mathrm{Cof}\,\nabla\psi_n\rightharpoonup H$ in $L^b(\mathcal{B})$ and $\det\nabla\psi_n\rightharpoonup \delta$ in $L^c(\mathcal{B})$, for 
\[
a\geq 2,\quad a^{-1}+b^{-1}\leq 1,\quad c\geq 1.
\]
Then
\[
H=\mathrm{Cof}\,\nabla\psi\ \text{ and }\ \delta=\det\nabla\psi\ \text{ a.e. on $\mathcal{B}$}.
\]
\end{Lemma} 

In the following, the letter $C$ will be used to denote various, possibly different, positive constants. Recall that
$\Theta_\psi(X,Y)=|\psi(X)-\psi(Y)|^{-1}$ is a.e. defined on $\mathcal{B}\times\mathcal{B}$, for all $\psi\in\mathscr{A}^{(i)}$, $i=1,2$.
\begin{Lemma}\label{esttheta}
For all $\psi\in\mathscr{A}^{(i)}$, $i=1,2$, and for all $0<\lambda<3$ we have
\[
\int_\mathcal{B}\int_\mathcal{B} \Theta_\psi(X,Y)^\lambda\, dX\, dY\leq C\left(\int_\mathcal{B} \det\nabla\psi(X)^{\frac{\lambda}{\lambda-6}}\,dX\right)^\frac{6-\lambda}{3}.
\]
\end{Lemma}
\begin{proof}
By the the change of variables formula~\eqref{changevariable},
\[
\int_\mathcal{B}\int_\mathcal{B} \Theta_\psi(X,Y)^\lambda\, dX\, dY=\int_{\R^3}\int_{\R^3}\frac{\mathbb{I}_{\psi(\mathcal{B})}(x)}{\det\nabla\psi(\psi^{-1}(x))}\frac{\mathbb{I}_{\psi(\mathcal{B})}(y)}{\det\nabla\psi(\psi^{-1}(y))}\,\frac{dx\,dy}{|x-y|^\lambda}.
\]
By  the Hardy-Littlewood-Sobolev inequality~\cite[Th.~4.3]{L}
\begin{align*}
\int_\mathcal{B}\int_\mathcal{B} \Theta_\psi(X,Y)^\lambda\, dX\, dY&\leq C\left\|\frac{\mathbb{I}_{\psi(\mathcal{B})}}{(\det\nabla\psi)\circ\psi^{-1}}\right\|_{L^\frac{6}{6-\lambda}(\R^3)}^2\\
&=C\left(\int_{\psi(\mathcal{B})(x)}\det\nabla\psi(\psi^{-1}(x))^\frac{6}{\lambda-6}\,dx\right)^{\frac{6-\lambda}{3}}\\
&=C\left(\int_\mathcal{B} \det\nabla\psi(X)^{\frac{\lambda}{\lambda-6}}\,dX\right)^\frac{6-\lambda}{3},
\end{align*}
and the proof is complete.
\end{proof}

Using Lemma~\ref{esttheta} with $\lambda=1$ and H\"older's inequality we obtain
\begin{equation}\label{estpot}
|E_\mathrm{pot}[\psi]|\leq C\left(\int_\mathcal{B}\frac{dX}{\det\nabla\psi(X)^{1/5}}\right)^{\frac{5}{3}}\leq C\left(\int_\mathcal{B}\frac{dX}{\det\nabla\psi(X)^{s}}\right)^{\frac{1}{3s}},
\end{equation}
where $s$ is the exponent given in (w3)$^{(i)}$.
Combining~\eqref{estpot} with (w3)$^{(i)}$ and the lower bound~\eqref{rhoref} on $\rho_\mathrm{ref}$ we have
\begin{equation}\label{lowbound}
I[\psi]
\geq \alpha \rho_0 
\int_\mathcal{B}\frac{dX}{\det\nabla\psi(X)^{s}}
-   \|\rho_\mathrm{ref}\|_{L^\infty(\mathcal{B})} \| h\|_{L^1(\mathcal{B}) } 
-C\left(\int_\mathcal{B}\frac{dX}{\det\nabla\psi(X)^s}\right)^{\frac{1}{3s}}.
\end{equation}
Since $s>1/3$, we infer that $I[\psi]$ is bounded from below over $\mathscr{A}^{(i)}$.
The lower boundedness of $I[\psi]$ implies the existence of minimizing sequences in the space $\mathscr{A}^{(i)}$, i.e. a sequence 
 denoted by $\{\psi_n^{(i)}\}$ such that 
 $$
 I[\psi_n^{(i)} ] \longrightarrow \inf_{\mathscr{A}^{(i)}} I , \quad \mbox{as $n\to\infty$}.
 $$
Owing to~\eqref{lowbound}, along any minimizing sequence we have
\begin{equation}\label{estdeterminant}
\int_\mathcal{B}\frac{dX}{\det\nabla\psi^{(i)}_n(X)^{s}}\leq C,
\end{equation}
and by~\eqref{estpot} we obtain that the sequence $\{E_\mathrm{pot}[\psi^{(i)}_n]\}$ is bounded.
Since the sequence $\{\psi^{(i)}_n\}$ is minimizing, then 
$\{E_\mathrm{str}[\psi^{(i)}_n]\}$ is also bounded. Using (w3)$^{(i)}$,~\eqref{rhoref} and the bound $\det F\leq C|F|^3$, we infer that
\begin{align*}
&\{\nabla\psi_n^{(i)}\}\text{ is bounded in $L^p(\mathcal{B})$},\\ 
&\{\det\nabla\psi_n^{(i)}\}\text{ is bounded in $L^{p/3}(\mathcal{B})$}, \\
&\{\mathrm{Cof}\,\nabla\psi_n^{(i)}\} \text{ is bounded in $L^q(\mathcal{B})$}.
\end{align*}


Finally, having prescribed the boundary value of $\{\psi^{(1)}_n\}\subset\mathscr{A}^{(1)}$ and the average of $\{\psi^{(2)}_n\}\subset\mathscr{A}^{(2)}$, Poincar\'e's inequality entails that
\[
\{\psi_n^{(i)}\}\text{ is bounded in $W^{1,p}(\mathcal{B})$}.
\]

It follows that any minimizing sequence possesses a subsequence  $\{\psi_n^{(i)}\}$ such that, for some $\psi^{(i)}\in W^{1,p}(\mathcal{B})$, $\delta\in L^{p/3}(\mathcal{B})$, $H\in L^q(\mathcal{B})$,
\[
\psi_n^{(i)}\rightharpoonup\psi^{(i)},\quad \text{in $W^{1,p}(\mathcal{B})$},\qquad \det\nabla\psi_n^{(i)}\rightharpoonup \delta\quad\text{in $L^{p/3}(\mathcal{B})$},\qquad\mathrm{Cof}\,\nabla\psi_n^{(i)}\rightharpoonup H\quad\text{in $L^{q}(\mathcal{B})$}
\] 
and in addition, by the Rellich-Kondra\v{s}ov imbedding theorem, 
\[
\psi_n^{(i)}\to \psi^{(i)},\quad\text{in $C^0(\overline{\mathcal{B}})$}.
\]
Moreover, by Lemma~\ref{tools},
\[
\delta=\det\nabla\psi^{(i)},\quad H=\mathrm{Cof}\,\nabla\psi^{(i)}.
\]
Now we show that 
\begin{equation}\label{closed}
\psi^{(i)}\in\mathscr{A}^{(i)}.
\end{equation}
By weak convergence, $\det\nabla\psi^{(i)}\geq 0$ a.e. on $\mathcal{B}$.
We claim that $\det\nabla\psi^{(i)}> 0$ a.e. on $\mathcal{B}$. Indeed, suppose by contradiction that there exists a set $L$ such that 
meas $L>0$ and $\det\nabla\psi^{(i)}= 0$ on $L$. 
Then, as shown in~\cite[pagg.~374-5]{C}, there exists a 
subsequence of $\{\psi^{(i)}_n\}$ (not relabeled) which converges pointwise a.e. on $L$ to a function $\psi^{(i)}$ satisfying $\det\nabla\psi^{(i)}= 0$ on $L$. 
Consequently, by Fatou lemma, 
$$\liminf_{n\to\infty} E_{\mathrm{str}}[\psi^{(i)}_n]=+\infty,$$
which implies that $I=+\infty$ for all $\psi\in\mathscr{A}^{(i)}$. However this is impossible, because the space $\mathscr{A}^{(i)}$ is not empty, and so
$\det\nabla\psi^{(i)}>0$ holds a.e. on $\mathcal{B}$. This completes the proof of~\eqref{closed} when $i=1$. Let us now consider the case $i=2$. Clearly $\psi^{(1)}$ satisfies the last condition in the definition of $\mathscr{A}^{(1)}$ (by the weak convergence $\psi_n^{(1)}\rightharpoonup\psi^{(1)}$ in $L^p(\mathcal{B})$) and therefore it only remains to show that meas$\,S_{\psi^{(1)}}=0$.  By~\eqref{changevariable}, the identity 
\[
\int_\mathcal{B}\det\nabla\psi^{(1)}_n(X)\,dX=\mathrm{meas}\,\psi^{(1)}_n(\mathcal{B}),
\]
holds for all minimizing sequences $\{\psi^{(1)}_n\}$.
Passing to the limit $n\to\infty$, using the weak convergence $\det\nabla\psi_n^{(1)}\rightharpoonup \det\nabla\psi^{(1)}$ in $L^{p/3}(\mathcal{B})$ and the uniform convergence $\psi_n^{(1)}\to \psi^{(1)}$, we obtain
\[
\int_\mathcal{B}\det\nabla\psi^{(1)}(X)\,dX=\mathrm{meas}\,\psi^{(1)}(\mathcal{B}).
\]
Applying~\eqref{changevariablegen} to the left hand side we obtain
\begin{equation}\label{temporale}
\int_{\psi^{(1)}(\mathcal{B})}\mathrm{card}\{(\psi^{(1)})^{-1}(x)\}dx=\mathrm{meas}\,\psi^{(1)}(\mathcal{B}).
\end{equation}
If meas$\,S_{\psi^{(1)}}>0$, we would have
\begin{align*}
\int_{\psi^{(1)}(\mathcal{B})}\mathrm{card}\{(\psi^{(1)})^{-1}(x)\}dx&=\mathrm{meas}\,(\psi^{(1)}(\mathcal{B})\setminus S_{\psi^{(1)}})+\int_{S_{\psi^{(1)}}}\mathrm{card}\{(\psi^{(1)})^{-1}(x)\}dx\\
&>\mathrm{meas}\,(\psi^{(1)}(\mathcal{B})\setminus S_{\psi^{(1)}})+\mathrm{meas}\,S_{\psi^{(1)}}=\mathrm{meas}\,\psi^{(1)}(\mathcal{B})
\end{align*}
and therefore~\eqref{temporale} leads to an absurd result. Hence  meas$\,S_{\psi^{(1)}}=0$ must hold, which completes the proof of~\eqref{closed} for $i=2$.

Now, the assumptions (w1), (w2) imply
\[
\liminf_{n\to\infty}\int_\mathcal{B}\rho_\mathrm{ref}(X)w(X,\nabla\psi_n)\,dX\geq \int_\mathcal{B}\rho_\mathrm{ref}(X)w(X,\nabla\psi^{(i)})\,dX,
\]
see~\cite{Ball} and~\cite[Th.~7.7-1]{C}.
Thus the proof of existence of minimizers is complete if we show that
\begin{equation}\label{convpot}
 E_\mathrm{pot}[\psi_n^{(i)}]\to E_\mathrm{pot}[\psi^{(i)}]. 
 \end{equation}
We estimate
\begin{align*}
|E_\mathrm{pot}[\psi_n^{(i)}]-&E_\mathrm{pot}[\psi^{(i)}]|\leq C\int_\mathcal{B}\int_\mathcal{B}|\Theta_{\psi_n^{(i)}}(X,Y)-\Theta_{\psi^{(i)}}(X,Y)|\,dX\,dY\\
&=C\int_\mathcal{B}\int_\mathcal{B}\Theta_{\psi_n^{(i)}}\Theta_{\psi^{(i)}}||\psi^{(i)}(X)-\psi^{(i)}(Y)|-|\psi_n^{(i)}(X)-\psi_n^{(i)}(Y)||\,dX\,dY\\
&\leq C\int_\mathcal{B}\int_\mathcal{B}\Theta_{\psi_n^{(i)}}\Theta_{\psi^{(i)}}|\psi^{(i)}(X)-\psi^{(i)}(Y)-\psi_n^{(i)}(X)+\psi_n^{(i)}(Y)|\,dX\,dY\\
&\leq C\int_\mathcal{B}\int_\mathcal{B}\Theta_{\psi_n^{(i)}}\Theta_{\psi^{(i)}}|\psi_n^{(i)}(X)-\psi^{(i)}(X)|\,dX\,dY\\
&\leq C\|\psi_n^{(i)}-\psi^{(i)}\|_{L^\infty(\mathcal{B})}\|\Theta_{\psi_n^{(i)}}\Theta_{\psi^{(i)}}\|_{L^1(\mathcal{B}\times \mathcal{B})}\\
&\leq C\|\psi_n^{(i)}-\psi^{(i)}\|_{L^\infty(\mathcal{B})}\|\Theta_{\psi_n^{(i)}}\|_{L^2(\mathcal{B}\times \mathcal{B})}\|\Theta_{\psi^{(i)}}\|_{L^2(\mathcal{B}\times \mathcal{B})}.
\end{align*}
By Lemma~\ref{esttheta},~\eqref{estdeterminant}, and since $s>1/2$, $\{\Theta_{\psi_n^{(i)}}\}$ and $\{\Theta_{\psi^{(i)}}\}$ are bounded in $L^{2}(\mathcal{B}\times\mathcal{B})$; using the uniform convergence $\psi_n^{(i)}\to\psi^{(i)}$, the claim~\eqref{convpot} follows.

To conclude our discussion on the existence of minimizers, we want to show that the convergence of the potential energy functional can also be obtained by using the properties of the spatial density along minimizing sequences, which is the argument used for fluids~\cite{GR}. We define
\[
\overline{\rho}_{\psi_{n}^{(i)}}=\rho_{\psi_n^{(i)}}\mathbb{I}_{\psi_n^{(i)}(\mathcal{B})}.
\]
Before stating the next result, we remark that since $\{\psi_n^{(i)}\}$ is uniformly bounded in $\mathcal{B}$, the supports of the spatial densities $\overline{\rho}_{\psi_n^{(i)}}$ are all contained in a common compact region of $\R^3$.

\begin{Proposition}\label{convrho}
For any minimizing sequence $\{\psi_n^{(i)}\}\subset\mathscr{A}^{(i)}$, there exists a subsequence such that
\[
\overline{\rho}_{\psi_n^{(i)}}\rightharpoonup\overline{\rho}_{\psi^{(i)}}\quad\text{in}\ L^\gamma(\R^3),\quad \text{for all }\,1\leq \gamma\leq 1+s. 
\]
Moreover  $e_\mathrm{pot}[\overline{\rho}_{\psi_n^{(i)}}]\to e_\mathrm{pot}[\overline{\rho}_{\psi^{(i)}}]$.
\end{Proposition}
\begin{proof}
First we observe that
\[
\int_{\R^3}\overline{\rho}_{\psi_n^{(i)}}(x)^{1+s}\,dx=\int_\mathcal{B}\frac{\rho_{\mathrm{ref}}(X)^{1+s}}{\det\nabla\psi_n^{(i)}(X)^s}\,dX\leq 
\|\rho_\mathrm{ref}\|_{L^\infty (\mathcal{B})}^{1+s}
\int_{\mathcal{B}}\frac{dX}{\det\nabla\psi_n^{(i)}(X)^s},
\]
whence by~\eqref{estdeterminant}, $\overline{\rho}_{\psi_n^{(i)}}$ is bounded in $L^{1+s}(\R^3)$. Since furthermore $\|\overline{\rho}_{\psi_n^{(i)}}\|_{L^1(\R^3)}=M$, for all $n$, there exists $\rho_*^{(i)}\in L^\gamma(\R^3)$ and a subsequence $\{\overline{\rho}_{\psi_n^{(i)}}\}$ such that
\[
\overline{\rho}_{\psi_n^{(i)}}\rightharpoonup\rho_*^{(i)}\quad\text{in}\ L^\gamma(\R^3),\ 1<\gamma\leq 1+s.
\]
Moreover by the Dunford-Pettis Theorem,
\[
\overline{\rho}_{\psi_n^{(i)}}\rightharpoonup\rho_*^{(i)}\quad\text{in}\ L^1(\R^3).
\]
Next we show that $\rho_*^{(i)}=\overline{\rho}_{\psi^{(i)}}$ almost everywhere. Let $\phi\in C^0(\R^3)$. We have
\begin{align*}
\int_{\R^3}(\rho_*^{(i)}-\overline{\rho}_{\psi^{(i)}})\phi\,dx&=\lim_{n\to\infty}\int_{\R^3}(\overline{\rho}_{\psi_n^{(i)}}-\overline{\rho}_{\psi^{(i)}})\phi\,dx\\
&=\lim_{n\to\infty}\int_\mathcal{B}\rho_\mathrm{ref}(X)[\phi(\psi_n^{(i)}(X))-\phi(\psi^{(i)}(X))] \,dX=0,
\end{align*}
and the claim follows. 
Finally, the assertion $e_\mathrm{pot}[\overline{\rho}_{\psi_n^{(i)}}]\to e_\mathrm{pot}[\overline{\rho}_{\psi^{(i)}}]$ follows by the well-known compactness properties of the functional $e_\mathrm{pot}$, see~\cite[pag.~125]{Lions} and~\cite[Lemma~3.7]{GR2}.
\end{proof}

\subsection{Proof of Theorem~\ref{propmin}}
The claim $\sigma^{(i)}\in L^1(\Omega^{(i)})$ follows by (w4) and the change of variable formula, for we have
\begin{align*}
\int_{\Omega^{(i)}}|\sigma^{(i)}(x)|\,dx&=\int_{\Omega^{(i)}}\left|\frac{\partial w}{\partial F}(X,\nabla\psi^{(i)})\cdot(\nabla\psi^{(i)})^T\right|_{_{X=(\psi^{(i)})^{-1}(x)}}\frac{\rho_\mathrm{ref}((\psi^{(i)})^{-1}(x))}{\det\nabla\psi^{(i)}((\psi^{(i)})^{-1}(x))} \,dx\\
&\leq\int_\mathcal{B}\left|\frac{\partial w}{\partial F}(X,\nabla\psi^{(i)})\cdot(\nabla\psi^{(i)})^T\right|dX\leq K\int_{\mathcal{B}}(1+w(X,\nabla\psi^{(i)}))\,dX<\infty.
\end{align*}
Next we show the validity of the identity~\eqref{viral}. Let $\tau<1$ and set
\[
\psi^{(1)}_\tau=(1-\tau)\psi^{(1)}-\tau\frac{a}{M}.
\]
Clearly, $\psi^{(1)}_\tau\in\mathscr{A}^{(1)}$, for all $\tau<1$. Moreover
\[
I[\psi^{(1)}_\tau]=\int_\mathcal{B}\rho_\mathrm{ref}(X)w(X,(1-\tau)\nabla\psi^{(1)})\,dX+\frac{1}{(1-\tau)}E_\mathrm{pot}[\psi^{(1)}]=\l(\tau).
\]
We claim that $\l(\cdot)$ is differentiable at $\tau=0$. To prove this we estimate, for $\tau>0$,
\begin{align*}
|w(X,\nabla\psi^{(1)}_\tau)-w(X,\nabla\psi^{(1)})|&\leq \left|\int_0^1\frac{d}{ds}\left(w(X,(1-s\tau)\nabla\psi^{(1)})\right)ds\right|\\
&\leq \tau\left|\int_0^1\tr\left(\frac{\partial w}{\partial F}(X,(1-s\tau)\nabla\psi^{(1)})\cdot(\nabla\psi^{(1)})^T\right)ds\right|\\
&\leq K\tau\left(1+\int_0^1 w(X,(1-s\tau)\nabla\psi^{(1)})ds\right)\\
&\leq K\tau(1+w(X,\nabla\psi^{(1)}))\\
&\quad+K\int_0^\tau|w(X,\nabla\psi_t^{(1)})-w(X,\nabla\psi^{(1)})|dt.
\end{align*}
Thus by Gr\"onwall's inequality
\[
|w(X,\nabla\psi^{(1)}_\tau)-w(X,\nabla\psi^{(1)})|\leq K\tau e^{K\tau} (1+w(X,\nabla\psi^{(1)})).
\]
The same estimate with $\tau$ replaced by $-\tau$ in the right hand side holds for $\tau<0$, whence
\[
|w(X,\nabla\psi^{(1)}_\tau)-w(X,\nabla\psi^{(1)})|\leq K|\tau| e^{K|\tau|}(1+w(X,\nabla\psi^{(1)})),\quad \forall\,\tau<1.
\]
It follows that
\begin{equation}\label{temporal}
\sup_{|\tau|\leq 1/2}\frac{1}{|\tau|}|w(X,\nabla\psi^{(1)}_\tau)-w(X,\nabla\psi^{(1)})|\leq K\,e^{K/2}(1+w(X,\nabla\psi^{(1)})).
\end{equation}
By the Dominated Convergence Theorem, $\l(\cdot)$ is differentiable at $\tau=0$ and we have
\[
\l'(0)=-\int_\mathcal{B}\rho_\mathrm{ref}(X)\tr\left(\frac{\partial w}{\partial F}(X,\nabla\psi^{(1)})\cdot(\nabla\psi^{(1)})^T\right)dX+E_\mathrm{pot}[\psi^{(1)}].
\]
Since $\psi^{(1)}$ is a minimizer, then $\l'(0)=0$ must hold:
\[
\int_\mathcal{B}\rho_\mathrm{ref}(X)\tr\left(\frac{\partial w}{\partial F}(X,\nabla\psi^{(1)})\cdot(\nabla\psi^{(1)})^T\right)dX=E_\mathrm{pot}[\psi^{(1)}].
\]
Upon the change of variable $x=\psi^{(1)}(X)$, the latter transforms into~\eqref{viral}.

Now let $\eta^{(1)}\in\mathcal{D}_{\Omega^{(1)}}(\R^3)$ and
\[
\eta^{(2)}:\overline{\Omega^{(2)}}\to\R^3,\quad \eta^{(2)}\in C^{1}(\overline{\Omega^{(2)}}),\quad\text{such that } \eta^{(2)}(x)=0,\ \text{for }x\in\partial\Omega^{(2)}.
\]
For $\tau\in\R$ we define
\[
\psi^{(i)}_\tau(X)=\psi^{(i)}(X)+\tau\,\eta(\psi^{(i)}(X))+(i-2)\frac{\tau}{M}\int_\mathcal{B}\rho_\mathrm{ref}\eta(\psi^{(i)}(X))dX.
\] 
We have
\[
\nabla\psi^{(i)}_\tau(X)=(I+\tau\nabla\eta(\psi^{(i)}(X)))\cdot\nabla\psi^{(i)}(X),
\]
hence
\[
\det\nabla\psi^{(i)}_\tau=\det\nabla\psi^{(i)}\det (I+\tau\nabla\eta(\psi^{(i)}(X))).
\]
Since $\det(I+\tau A)=1+\tau\tr(A)+O(\tau^2)$, there exists $\tau_0>0$ such that $\det\nabla\psi^{(i)}_\tau>0$ a.e. on $\mathcal{B}$ for $|\tau|\leq\tau_0$. Clearly $\psi^{(2)}_\tau\in\mathscr{A}^{(2)}$, for $|\tau|\leq\tau_0$. Moreover
\[
\int_\mathcal{B}\rho_{\mathrm{ref}}\psi^{(1)}_\tau(X)\,dX=a.
\]
Thus $\psi^{(1)}_\tau\in\mathscr{A}^{(1)}$ follows if we show that meas$\,S_{\psi^{(1)}_\tau}=0$ for $|\tau|$ small enough. To this purpose we need the following result.
\begin{Proposition}
There exists $C>0$ such that
\begin{equation}\label{holder}
|\eta^{(i)}(x)-\eta^{(i)}(y)|\leq C|x-y|,\quad\text{for all }x,y\in\Omega^{(i)},
\end{equation}
i.e., $\eta^{(i)}\in C^{0,1}(\Omega^{(i)})$.
\end{Proposition}
\begin{proof}
We shall use the following property of $C^1$ functions $f:\Omega\to\R^3$ defined on a regular domain $\Omega$: There exists a constant $C$, which depends only on $\Omega$, such that $|f(x)-f(y)|\leq C\|\nabla f\|_{L^\infty(\Omega)}|x-y|$, for all $x,y\in\Omega$. When $\Omega$ is convex this is a trivial consequence of the Mean Value Theorem for vector valued functions (in this case we may choose $C=1$); for general domains $\Omega$, the claim is proved as follows. A standard property of any domain $\Omega$ (see~\cite[pag.~224]{C}) is that, for all $x,y\in\overline{\Omega}$, there exists a finite sequence of points $z_1,\dots z_{N+1}$, such that $z_1=x$, $z_{N+1}=y$, $z_i\in\Omega$, for all $i=2,\dots N$, the open segment with end points $z_i,z_{i+1}$ is all contained in $\Omega$,  and the inequality 
\[
\sum_{i=1}^{N}|z_i-z_{i+1}|\leq C|x-y|
\]
holds for a positive constant $C$ which depends only on $\Omega$. Using this property we have
\begin{align*}
|f(x)-f(y)|&=\left|\sum_{i=1}^N f(z_i)-f(z_{i+1})\right|=\left|\sum_{i=1}^N\int_0^1\frac{d}{ds}[f(s z_i+(1-s)z_{i+1})]\,ds\right|\\
&\leq \|\nabla f\|_{L^\infty(\Omega)}\sum_{i=1}^{N}|z_i-z_{i+1}|\leq C\|\nabla f\|_{L^\infty(\Omega)}|x-y|.
\end{align*}
Since $\Omega^{(2)}$ is a regular domain, the inequality~\eqref{holder} for $i=2$ is proved. 
For $i=1$, let $\phi\in C^1_c(\R^3)$ such that $\phi_{|_{\Omega^{(1)}}}=\eta^{(1)}$ and let $R>0$ such that $\Omega^{(1)}\subset B_R$, where $B_R$ denotes the ball with center in the origin and radius $R$. Applying the preceding result we have
\[
|\phi(x)-\phi(y)|\leq C|x-y|,\quad\text{for all $x,y\in B_R$},
\] 
where $C=C(R)$. Restricting the previous inequality to $x,y\in\Omega^{(1)}$ yields~\eqref{holder} for $i=1$.
\end{proof}
\begin{Corollary}
There exists $\tau_1>0$ such that $K_{\psi^{(1)}_\tau}=K_{\psi^{(1)}}$, for $|\tau|\leq\tau_1$. In particular, $\mathrm{meas}\,K_{\psi^{(1)}_\tau}=0$ and thus
\[
\mathrm{meas}\,S_{\psi^{(1)}_\tau}=\mathrm{meas}\,\psi^{(1)}_\tau(K_{\psi^{(1)}_\tau})=0,\quad\text{for $|\tau|\leq\tau_1$}.
\]
\end{Corollary}
\begin{proof}
Clearly, $K_{\psi^{(1)}}\subseteq K_{\psi^{(1)}_\tau}$. Now, let $X\in K_{\psi^{(1)}_\tau}$. Thus there exists $Y\in\mathcal{B}$ such that $Y\neq X$ and $\psi^{(1)}_\tau(X)=\psi^{(1)}_\tau(Y)$, i.e,
\[
\psi^{(1)}(X)+\tau\,\eta(\psi^{(1)}(X))=\psi^{(1)}(Y)+\tau\,\eta(\psi^{(1)}(Y)).
\]
Using~\eqref{holder} we obtain
\[
|\psi^{(1)}(X)-\psi^{(1)}(Y)|\geq|\psi^{(1)}(X)-\psi^{(1)}(Y)|(1-C\tau). 
\] 
Thus $\psi^{(1)}(X)=\psi^{(1)}(Y)$, for $|\tau|\leq (2C)^{-1}=\tau_1$; hence $X\in K_{\psi^{(1)}_\tau}$ and $K_{\psi^{(1)}_\tau}\subseteq K_{\psi^{(1)}}$. 
\end{proof}
In the following we shall assume that $|\tau|\leq\min\{\tau_0,\tau_1\}$, so that $\psi^{(i)}_\tau\in\mathscr{A}^{(i)}$.  Moreover, we simplify the notation by denoting both $\eta^{(1)}$ and $\eta^{(2)}$ by $\eta$. We continue by proving that the function $\tau\rightarrow I[\psi^{(i)}_\tau]$ is differentiable at $\tau=0$. We have
\begin{equation}\label{temp}
\lim_{\tau\to 0}\frac{I[\psi^{(i)}_\tau]-I[\psi^{(i)}]}{\tau}=\lim_{\tau\to 0}\frac{1}{\tau}(E_\mathrm{str}[\psi^{(i)}_\tau]-E_\mathrm{str}[\psi^{(i)}])
+\lim_{\tau\to 0}\frac{1}{\tau}(E_\mathrm{pot}[\psi^{(i)}_\tau]-E_\mathrm{pot}[\psi^{(i)}]).
\end{equation}
As shown by J.~Ball in~\cite[pag.~13]{ball2}, the bound (w4) implies that the function $\tau\to E_\mathrm{str}[\psi^{(i)}_\tau]$ is differentiable at $\tau=0$ and
\begin{align*}
\lim_{\tau\to 0}\frac{1}{\tau}(E_\mathrm{str}[\psi^{(i)}_\tau]-E_\mathrm{str}[\psi^{(i)}])&=\left(\frac{d}{d\tau}E_\mathrm{str}(\psi^{(i)}_\tau)\right)_{\tau=0}=\int_\mathcal{B}\left(\frac{d}{d\tau}w(X,\nabla\psi^{(i)}_\tau)\right)_{\tau=0}\rho_\mathrm{ref}\,dX\\
&=\int_\mathcal{B}\mathrm{Tr}\left(\frac{\partial w}{\partial F}(X,\nabla\psi^{(i)})\cdot(\nabla\psi^{(i)})^T\cdot(\nabla\eta(\psi^{(i)}(X)))^T\right)\rho_\mathrm{ref}\,dX.
\end{align*}
It will now be shown that a similar property holds for the potential term in~\eqref{temp}. By~\eqref{holder},
\[
|\psi^{(i)}_\tau(X)-\psi^{(i)}_\tau(Y)-(\psi^{(i)}(X)-\psi^{(i)}(Y))|\leq C |\tau||\psi^{(i)}(X)-\psi^{(i)}(Y)|;
\]
thus the estimate 
\[
|\psi^{(i)}_\tau(X)-\psi^{(i)}_\tau(Y)|\geq |\psi^{(i)}(X)-\psi^{(i)}(Y)|/2,\quad \text{i.e.,}\quad \Theta_{\psi^{(i)}_\tau}(X,Y)\leq 2\Theta_{\psi^{(i)}}(X,Y),
\]
holds for $|\tau|$ small enough. 
Moreover
\begin{align*}
|\Theta_{\psi^{(i)}_\tau}(X,Y)-\Theta_{\psi^{(i)}}(X,Y)|&=\Theta_{\psi^{(i)}_\tau}(X,Y)\Theta_{\psi^{(i)}}(X,Y)||\psi^{(i)}_\tau(X)-\psi^{(i)}_\tau(Y)|-|\psi^{(i)}(X)-\psi^{(i)}(Y)||\\
&\leq \Theta_{\psi^{(i)}_\tau}(X,Y)\Theta_{\psi^{(i)}}(X,Y)|(|\psi^{(i)}_\tau(X)-\psi^{(i)}(X)|+|\psi^{(i)}_\tau(Y)-\psi^{(i)}(Y)|)\\
&\leq C|\tau| \Theta_{\psi^{(i)}_\tau}(X,Y)\Theta_{\psi^{(i)}}(X,Y).
\end{align*}
Thus for $|\tau|$ small enough we have
\[
\frac{1}{|\tau|}|\Theta_{\psi^{(i)}_\tau}(X,Y)-\Theta_{\psi^{(i)}}(X,Y)|\leq C\Theta_{\psi^{(i)}}(X,Y)^2.
\]
By Lemma~\ref{esttheta}, and since $s>1/2$, the function in the right hand side is integrable on $\mathcal{B}\times \mathcal{B}$. Then, by the Dominated Convergence Theorem, and since
\[
\lim_{\tau\to 0}\frac{1}{\tau}\left(\Theta_{\psi^{(i)}_\tau}(X,Y)
-\Theta_{\psi^{(i)}}(X,Y)\right)=-\frac{(\psi^{(i)}(X)-\psi^{(i)}(Y))\cdot(\eta(\psi^{(i)}(X))-\eta(\psi^{(i)}(Y)))}{|\psi^{(i)}(X)-\psi^{(i)}(Y)|^3},
\]
for almost all $X,Y\in \mathcal{B}$, we obtain
\begin{align*}
\lim_{\tau\to 0}\frac{1}{\tau}(E_\mathrm{pot}[\psi^{(i)}_\tau]-E_\mathrm{pot}[\psi^{(i)}])
&=-
\lim_{\tau\to 0}\frac{1}{2\tau} 
\int_\mathcal{B}\int_\mathcal{B}\rho_\mathrm{ref}(X)\rho_\mathrm{ref}(Y)\left(\Theta_{\psi^{(i)}_\tau}(X,Y)
-\Theta_{\psi^{(i)}}(X,Y)\right)\,dX\, dY
\\
&=\int_\mathcal{B}\int_\mathcal{B}\rho_\mathrm{ref}(X)\rho_\mathrm{ref}(Y)\frac{(\psi^{(i)}(X)-\psi^{(i)}(Y))\cdot\eta(\psi^{(i)}(X))}{|\psi^{(i)}(X)-\psi^{(i)}(Y)|^3}\,dX\, dY.
\end{align*}
Substituting into~\eqref{temp} and using that $\psi^{(i)}$ is a minimizer yields
\begin{align*}
&\int_\mathcal{B}\mathrm{Tr}\left(\frac{\partial w}{\partial F}(X,\nabla\psi^{(i)})\cdot(\nabla\psi^{(i)}(X))^T\cdot(\nabla\eta(\psi^{(i)}(X)))^T\right)dX\\
&\qquad+\int_\mathcal{B}\int_\mathcal{B}\rho_\mathrm{ref}(X)\rho_\mathrm{ref}(Y)\frac{(\psi^{(i)}(X)-\psi^{(i)}(Y))\cdot\eta(\psi^{(i)}(X))}{|\psi^{(i)}(X)-\psi^{(i)}(Y)|^3}\,dX\, dY=0.
\end{align*}
Upon the change of variable $x=\psi^{(i)}(X)$, $y=\psi^{(i)}(Y)$, the latter equation transforms into~\eqref{equation2}, and the proof of Theorem~\ref{propmin} is concluded.

\begin{center}
{\bf Acknowledgements}
\end{center}
The first author would like to thank Jes\'us Montejo for many helpful discussions, as well as Robert~Beig and Bernd~Schmidt for their important comments on a previous version of the paper. The second author has been supported by the ``Ministerio de Ciencia e Innovaci\'on" (MICINN) of Spain (Proj. MTM2009--10878) and ``Junta de Andaluc\'{\i}a" (Proj. FQM--116).


\end{document}